\documentclass{lmcs}
\usepackage{amssymb}
\usepackage{algorithm}
\usepackage{algpseudocode}
\usepackage{booktabs}
\setlist[enumerate]{leftmargin=25pt}
\usepackage{xypic}

\newtheorem{exam}[thm]{Example}
\newtheorem{remark}[thm]{Remark}

\newenvironment{theorem}[1]{\begin{thm} \rm \label{thm:#1} }{\end{thm}}
\newenvironment{definition}[1]{\begin{defi} \rm \label{df:#1} }{\end{defi}}
\newenvironment{lemma}[1]{\begin{lem} \rm \label{lem:#1} }{\end{lem}}
\newenvironment{corollary}[1]{\begin{cor} \rm \label{cor:#1} }{\end{cor}}
\newenvironment{proposition}[1]{\begin{prop} \rm \label{pr:#1} }{\end{prop}}

\newcommand{\df}[1]{Definition~\ref{df:#1}}

\newcommand{\Sec}[1]{Section~\ref{sec:#1}}
\newcommand{\plat}[1]{\raisebox{0pt}[0pt][0pt]{#1}}     
\newcommand{\lock}{{\tt lock}}
\newcommand{\unlock}{{\tt unlock}}
\newcommand{\MUTEX}{{\sc mutex}}
\DeclareMathAlphabet{\mathbbm}{U}{bbm}{m}{n}            
\DeclareSymbolFont{frenchscript}{OMS}{ztmcm}{m}{n}
\DeclareMathSymbol{\Pow}{\mathord}{frenchscript}{80}    

\begin{document}

\title[]
    {On the Notions of Bounded Bypass, \linebreak and How to Make any Deadlock-Free MUTEX Protocol Satisfy One of Them}
\author[]{Rob van Glabbeek}[a]
\address{School of Informatics, University of Edinburgh, UK\newline
School of Computer Science and Engineering, University of New South Wales, Sydney, Australia}
\thanks{Supported by Royal Society Wolfson Fellowship RSWF\textbackslash R1\textbackslash 221008}
\email{rvg@cs.stanford.edu}
\author[]{Daniele Gorla}[b]
\address{Dept.\ of Computer Science, Sapienza University of Rome, IT}
\email{gorla@di.uniroma1.it}
\author[]{Myrthe Spronck}[c]
\address{Dept.\ of Mathematics and Computer Science, Eindhoven University of Technology, The Netherlands}
\email{m.s.c.spronck@tue.nl}
\thanks{We thank Bas Luttik for suggestions on the text and Tim Willemse for improvements to the modal $\mu$-calculus formulas.}


\begin{abstract}
In the literature on mutual exclusion, bounded bypass has been used for a long time as a strengthening of starvation-freedom, but, to the best of our knowledge, it still lacks a satisfying definition as a liveness property on its own. Moreover, we have encountered {\MUTEX} protocols for which this notion needs to be slightly weakened in order to be met.

To solve these issues, we first provide a formal definition of bounded bypass (that also corrects a previous definition from Raynal) and then introduce the notions of {\em post-doorway} and {\em intermittent} bounded bypass, two liveness properties that lie between starvation-freedom and bounded bypass. Essentially, intermittent bounded bypass weakens bounded bypass by ignoring the possible bypasses that may happen during the execution of a certain finite set of write operations to shared registers. Orthogonally, post-doorway bounded bypass ignores the bypasses that may happen during a finite initial phase of the \lock\ protocol.

Furthermore, we study an algorithm proposed by Yoah Bar-David in 1998 to enhance the liveness properties of any deadlock-free {\MUTEX} protocol and prove that: (1) in the setting of atomic registers, this algorithm upgrades any deadlock-free mutual exclusion protocol to a bounded bypass one, with a bound that is quadratic in the number of processes; and (2) in the setting of safe and regular registers, the very same algorithm ensures the intermittent version of bounded bypass, still with a quadratic (but slightly different) bound.

Finally, we provide logical formulae for the different notions of bounded bypass defined in this paper and use them to confirm all claims made here, by using model checking. This had a positive impact on the theoretical development of the work, since it  allowed us to identify and correct small mistakes/ambiguities in definitions and proofs.

\end{abstract}

{\hfuzz 4pt
  \maketitle
}


\section{The {\MUTEX} Problem}

The MUTual EXclusion problem ({\MUTEX}, for short) \cite{Dijk65} is one of the most fundamental problems in concurrent and distributed systems. So, it is not surprising that many protocols have been devised in the last 60 years to solve it \cite{Ala03,And93,Ara11,BL93,BDH16,Kes82,Knu66,LamportMutex2,Pet81,SWPL11,Szy88,Szy90} (just to cite a few).
Its general formulation assumes a set of $n$ ($\geq$ 1) processes $p_0, p_1, \ldots p_{n-1}$ running in parallel, each sequentially executing its code. 
Processes are {\em reliable} and {\em asynchronous}. ‘‘Reliable'' means that each process correctly executes the code of the corresponding algorithm. ‘‘Asynchronous'' means that there is no assumption on the time it takes for a process to proceed from a state to the next one (which means that an asynchronous process proceeds at an arbitrary speed).

Each process has its own local memory (whose locations are denoted with lower case letters in what follows), but they can all access some shared memory to interact (whose locations are denoted with upper case letters). 
The memory units that a process can access are {\em registers}; 
the {\em domain} of a register is the set of all the values it can store (e.g., if the register is a sequence of $k$ bits, then its domain is the set of all $2^k$ binary strings of length $k$).
Every process has portions of code (called  {\em critical sections}) that must be run in isolation to avoid inconsistencies; these usually involve writing into the shared memory.

The {\MUTEX} problem is then formulated as follows:

\begin{definition}{}
\label{def:mutex}
Devise a protocol that ensures that at most one process at a time can be in a critical section.
\end{definition}

\noindent
This is usually accomplished by defining two functions: \lock, to be run to get access to a critical section, and \unlock, to be run when leaving the critical section. 

The main correctness property of the {\MUTEX} problem formulated above is a safety property; it is well-known that, to be meaningful, every safety property must be joint with a liveness property \cite{AS85}. 
Two fundamental liveness properties for {\MUTEX} protocols are the following:

\begin{definition}{Deadlock-freedom}
A {\MUTEX} protocol is {\em deadlock-free} if, whenever at some time $t$ at least one process is running its \lock, then afterwards a process will enter its critical section. (This needs not be one of the processes executing their \lock\ at time $t$.)
\end{definition}

\begin{definition}{Starvation-freedom}
A {\MUTEX} protocol is {\em starvation-free} whenever every invocation of a \lock\  eventually terminates (and lets the process that invoked the \lock\ enter its critical section).
\end{definition}

\noindent
A third property that has been used in the literature for a long time is {\em bounded bypass}. A discussion of this notion and its formalization will be presented in Section~\ref{sec:bypass}.

\section{Assumptions on the Hardware}

In order to reason on the correctness of {\MUTEX} protocols, it is vital to make explicit any assumption on the hardware on which the shared memory is operating. A central assumption made by Dijkstra in the first published paper on mutual exclusion \cite{Dijk65} is the following:
\begin{quote}
\emph{Writing a word into or nondestructively reading a word from this store are undividable operations; i.e., when two or more computers try to communicate (either for reading or for writing) simultaneously with the same common location, these communications will take place one after the other, but in an unknown order.}
\end{quote}
This assumption is sometimes called \emph{atomicity}.
When we acknowledge (as Dijkstra did) that read and write actions take time, atomicity implies that, if process $i$ tries to read or write to a given register while another process $j$ is already busy reading or writing to that register, process $i$ has to wait until process $j$ has completed its action. Thus a read or write action of one process may be postponed by a read or write action of another; a situation that is sometimes referred to as a \emph{blocking} \cite{CDV09}.

In most works on mutual exclusion, including \cite{Raynal,T} and this paper,
it is assumed instead that reads and writes to the shared memory are {\em non-blocking}\/; if a process tries to read or write, this happens immediately. When acknowledging that reads and writes take time, it follows that different actions on the same register may overlap in time. In this setting, there exits a hierarchy of memory models, based on the values which may be returned by read actions: we say that model M1 is \emph{at least as strong} as M2 if, for any given read, the set of values that could be returned according to M1 is contained in the one according to M2.

Two other common assumptions on the underlying hardware are that only read and write operations are available (i.e., test-and-set, compare-and-swap, or similar are not considered), and that they always terminate.
According to whether a register can be written/read by just one process or by many different ones, we have: single-write/single-read (SWSR), single-write/multiple-read (SWMR), multiple-write/single-read (MWSR), or multiple-write/multiple-read (MWMR) registers.
Lamport \cite{Lam86d} proposed three memory models for SWMR registers, which form a strict hierarchy; from the weakest to the strongest, these are the \emph{safe registers}, the \emph{regular registers} and the \emph{atomic registers}. In all these models, a read of a register $R$ must always return a value from the domain of $R$. Whereas there are many different ways to generalize the concept of SWMR regular register to MWMR registers \cite{SWPL11,spronck2023process}, there are canonical definitions of safe and atomic MWMR registers, which appear, e.g., in \cite{Raynal}.

\begin{definition}{Safe Register}
  A register $R$ is \emph{safe} if it guarantees the following:
  \begin{itemize}
  \item Initially, and whenever a write on $R$ terminates when no other write on $R$ is going on, $R$ contains a well-defined value---the value \emph{of} $R$---that does not change until a successive write commences.
  \item When a write to $R$ of value $v$ terminates, and this write had no overlap with any other write on $R$, then $v$ becomes the value of $R$.
  \item A read that does not overlap with any write returns the value of $R$.
  \end{itemize}
\end{definition}

\noindent
When we merely know that a register is safe, nothing more is guaranteed than the above. Thus a read that overlaps with a write may return any value in the domain of the register, and when the last of a set of partially overlapping writes terminates, the register may hold any value in its domain. 

We consider regular registers only in situations where overlapping writes are excluded; this is enough since the algorithm we are going to consider in Section~\ref{sec:alg} has SWMR registers and one MWMR register that is written only in mutual exclusion mode (to be precise, at the beginning of the \unlock).
In that case, a safe register $R$ is {\em regular} if it also satisfies the following requirement:
  \begin{itemize}
  \item A read that overlaps with a set of writes returns either any of the values written by these writes or the value contained in the register before the first of these writes started.
  \end{itemize}

\begin{definition}{Atomic Register}
  A register $R$ is \emph{atomic} if it behaves \emph{as if}, for each read and write on $R$, there is a single point in time (its \emph{execution point}) between its start and termination, when the action takes effect. For a write action of value $v$, this means that $v$ becomes the value of $R$ at the execution point; for a read action, this means that it will return the value that $R$ has at that time.
  Finally, different actions have different execution points, and the value of $R$ changes only at the execution points of write actions.
\end{definition}
\noindent
The above definition does not stipulate that an atomic register actually features such execution points. It merely says that, for the sole purpose of determining which value may be returned by a read action, we may assume it does.

\section{Bounded Bypass}
\label{sec:bypass}

The property of {\em bounded bypass} (sometimes referred to as {\em bounded waiting}) has been used for a long time as a desirable property of {\MUTEX} protocols. It is encountered in two variants: as a safety property\footnote{Or rather, a property that becomes a safety property when requiring a concrete bound.} that needs to be required in conjunction with deadlock- or starvation-freedom, or as a liveness property that is intended to imply deadlock- and starvation-freedom. The former use occurs in \cite[Sect.~5.2]{OSconcepts} (where it is paired with a form of deadlock-freedom called {\em progress}), with the following definition:
\begin{quote}\it
There exists a bound on the number of times
that other processes are allowed to enter their critical sections after a process has made a request to enter its critical section and before that
request is granted.
\end{quote}
It also occurs in \cite[pg.~48]{T}:
\begin{quote}\it
  If a given process is in its entry code, then there is a bound on the number of times any other process is able to enter its critical section before the given process does so.
\end{quote}
The latter use, as a liveness property that is intended to imply deadlock- and starvation-freedom, occurs in \cite[pg.~11]{Raynal}, where bounded bypass is defined as
\begin{quote}\it
There is a function $f(n)$ such that, each time a process invokes \lock, it loses at most $f(n)$ competitions with respect to the other processes.
\end{quote}
with $n$ the number of processes, and competition losing is defined as follows \cite[pg.~9]{Raynal}:
\begin{quote}\it
When several processes are simultaneously executing \lock, we say that they are competing for access to the critical section code. If one of these invocations terminates while the other invocations do not, the corresponding process $p$ is called the winner, while each other competing process $q$ is a loser (its invocation remains pending).
\end{quote}
Raynal claims that deadlock-freedom, starvation-freedom and bounded bypass (as defined above) form a ‘‘hierarchy of liveness properties: bounded bypass $\Rightarrow$ starvation-freedom $\equiv$ finite bypass $\Rightarrow$ deadlock-freedom" \cite[pg.~11]{Raynal}. This is however false, since 
nothing guarantees that whoever invokes \lock\ eventually enters its critical section in a {\MUTEX} protocol that satisfies bounded bypass (as defined above). As an extreme example, consider a protocol in which each {\lock} contains the single instruction ``\textbf{await} \textit{false}''.\footnote{The instruction \textbf{await} \textit{COND} blocks the invoking process until the boolean condition \textit{COND} becomes true; this can be implemented, e.g., by using busy waiting, but our results are not influenced by the actual implementation.} Such a protocol trivially satisfies bounded bypass as formulated in \cite{Raynal}, since no competition will ever be lost (nor will one ever be won); yet it does not satisfy deadlock- or starvation-freedom. This breaks the chain of implications reported above. Note, however, that the conjunction of deadlock-freedom and Raynal's bounded bypass does imply starvation-freedom.

We now correct the definition of Raynal, with the intention of having one single liveness property that implies starvation-freedom (and deadlock-freedom). To this aim, we say that a process $i$ running \lock\ \emph{loses} a competition with another process $j$ if the \lock\ of $j$ overlaps in time with that of $i$ and terminates earlier.
A first attempt could be the following formulation:

\begin{quote}\hfuzz 4pt
There is a function $f(n)$ such that, each time a process invokes \lock, \underline{it will} \underline{eventually indicate its desire to enter the critical section and, following this,} it enters after losing at most $f(n)$ competitions.
\end{quote}

This proposal fixes two mistakes in the definition provided by \cite{Raynal}. First, with the explicit requirement that the invoking process enters the critical section, the chain of implications reported at the end of \cite[Sect.~1.3.1]{Raynal} now holds. 
Second, the underlined part is missing in \cite{Raynal}, and also in \cite{T}, but must be intended anyway as, without it, bounded bypass is not a useful concept. 
Namely, as long as a process $p_i$ that has invoked its $\lock$ has not in any way communicated to the competing processes its intention to enter the critical section, from the perspective of these other processes $p_i$ could just as well have not invoked its \lock. For this reason, the other processes cannot hold back in favor of $p_i$, and so can bypass $p_i$ as many times as they want.
One might think that this issue can be mitigated by $p_i$ indicating its desire to enter the critical section in the very first instruction of its \lock, for instance by setting a flag that can be read by the other processes. However, even when assuming atomic registers, this first action takes time, so there exists a period during which the other processes are unaware of $p_i$’s desire and can bypass $p_i$ an unbounded number of times.
To avoid this pitfall, in the spirit of the definition of Silberschatz, Galvin and Gagne \cite{OSconcepts}, the underlined part needs to be added to Taubenfeld's and Raynal's definitions, as done in our proposal.

The above formulation is however quite unsatisfactory because it leaves unspecified what it really means for a process to indicate its desire (or to make a request) to enter the critical section. In many protocols in the literature, this is quite uncontroversial: it is accomplished by setting an SWMR variable (typically called {\em flag}) and this usually happens at the very beginning of the \lock. However, there are notable exceptions: the most remarkable one is Dekker's algorithm (reported for the first time as a flow chart in \cite{dekker1} and in pseudocode in \cite{Ala03}), one of the first and most famous {\MUTEX} protocols for two processes. We report it here as Algorithm~\ref{Dekker} to ease references in what follows.

\begin{algorithm}[h]
  \caption{Dekker's Algorithm (as reported in \cite{Ala03} -- it first appeared in \cite{dekker1}). Line 1 is the preliminary initialization; lines 3 and 11 are the way in which $p_i$ invokes \lock/\unlock, for $i=0$ or $1$.}
  \label{Dekker}
  \begin{algorithmic}[1]

   \State{Initialize FLAG[0] and FLAG[1] to 0, and TURN to either 0 or 1}

   \State{\ }

   \State{{\lock}($i$) := }
   \State\hspace{\algorithmicindent}{FLAG[$i$] $\leftarrow$ 1}
   \State\hspace{\algorithmicindent}{\bf while} FLAG[1$-i$] = 1 {\bf do}
   \State\hspace{\algorithmicindent}\hspace{\algorithmicindent}{{\bf if} TURN = 1$- i$ {\bf then}}
   \State\hspace{\algorithmicindent}\hspace{\algorithmicindent}\hspace{\algorithmicindent}{FLAG[$i$] $\leftarrow$ 0}
   \State\hspace{\algorithmicindent}\hspace{\algorithmicindent}\hspace{\algorithmicindent}{{\bf await} TURN = $i$}
   \State\hspace{\algorithmicindent}\hspace{\algorithmicindent}\hspace{\algorithmicindent}{FLAG[$i$] $\leftarrow$ 1}
   
   \State{\ }

   \State{{\unlock}($i$) := }
   \State\hspace{\algorithmicindent}{TURN $\leftarrow$ 1$-i$}
   \State\hspace{\algorithmicindent}{FLAG[$i$] $\leftarrow$ 0}
   
\end{algorithmic}
\end{algorithm}  

Here, the problem is that the flag (that is indeed set at the very beginning of the \lock) may change its value during the execution of the  \lock. This illustrates our dissatisfaction with the vagueness of phrases like ‘‘indicate its desire to enter the critical section'' or ``made a request to enter its
critical section'': which is the real moment in Dekker's algorithm when a process indicates its desire (or made a request) to enter the critical section?

One possible answer is: the moment in which the flag does not change anymore. In this way, Dekker's algorithm would satisfy bounded bypass. However, with this understanding, we could add to protocols an instruction that generates this very last setting just before entering the critical section; in this case, any starvation-free protocol can be turned into one that satisfies bounded bypass with bound 0! 

This leads us to move to the opposite extreme and allow a process to declare its intention with the first instruction of its \lock; we remark that this should be a write to a shared register, otherwise no form of bounded bypass is possible. 
We take the latter approach and provide our  definition of bounded bypass:

\begin{definition}{Bounded bypass}
A {\MUTEX} protocol satisfies {\em bounded bypass} whenever there is a function $f(n)$ such that every process that invokes the \lock\ enters the critical section, and this happens after losing at most $f(n)$ competitions from the moment in which it completes the execution of the first instruction of the \lock.
\end{definition}

Essentially, this definition provides a period of time in which many bypasses may happen (namely, the time a process takes to perform its first instruction); however, after this period, the number of possible bypasses is upper bounded by $f(n)$. Moreover, this aligns with virtually all {\MUTEX} protocols found in the literature, where the desire of entering the critical section is declared in the very first instruction of their \lock.

\begin{remark}
\label{rem:DekkerBB}
With this definition, Dekker's algorithm does {\em not} satisfy bounded bypass. 
To see this, consider the following execution:
\begin{itemize}
\item assume that TURN is initialized to 0;
\item $p_0$ and $p_1$ both simultaneously do lines 4, 5 and 6;
\item $p_0$ finds TURN = 0, so jumps back to line 5;
\item $p_1$ finds TURN = 0, so proceeds to line 7, executes it and then sleeps for a long time;
\item now, from $p_0$'s point of view, $p_1$ is like if it has not locked at all; so, $p_0$ can enter and exit the critical section as many times as it wants (until $p_1$ wakes up).
\end{itemize}
In contrast, it is well known that, when assuming atomic registers, Dekker's algorithm satisfies starvation-freedom (see \cite{GLS25} for an automated proof of this property).
\end{remark}

 \subsection{Post-Doorway Bounded Bypass}\mbox{}
 \label{sec:pdbb}

\bigskip
\noindent
\df{Bounded bypass} proposes a definition of bounded bypass that fixes many issues of other definitions proposed in the literature.
However, we realize that this is not the only possible such definition for the intended property.
We explore here a first possible variant, based on an observation made by Lamport in \cite{LamportMutex2}.
Essentially, he assumes that every \lock\ protocol is divided in two parts, which he calls {\em doorway} and {\em waiting};
the doorway is a portion of code whose execution requires only a bounded number of elementary operation executions (and hence always terminates).
Thus, a plausible alternative to \df{Bounded bypass} consists in counting only the bypasses that happen during the waiting part, that is, after the doorway part has been completed. This leads us to the following notion of bounded bypass, that we call {\em post-doorway}.\footnote{\cite{LamportMutex2} reports something very similar to the following definition, calling it {\em $r$-bounded waiting} (see the last paragraph of \cite[Sect.2.2]{LamportMutex2}) and ascribing it to Rivest and Pratt \cite{RP76}. However, the definition in \cite{RP76} does not mention the doorway at all; moreover, it only poses a linear bound on how many times every process may bypass a given one, whereas the following definition aligns with Raynal and assumes a global bound that can be any function of $n$.}

\begin{definition}{Doorway bounded bypass}
A {\MUTEX} protocol satisfies {\em post-doorway bounded bypass} whenever there is a function $f(n)$ and an initial portion of the \lock\ whose execution requires only a bounded number of elementary actions (the {\em doorway}) such that every process that invokes the \lock\ enters the critical section, and this happens after losing at most $f(n)$ competitions from the moment in which it completes the doorway.
\end{definition}

\noindent
Clearly, \df{Doorway bounded bypass} weakens \df{Bounded bypass}: if we have a protocol that satisfies the latter, then the protocol also satisfies the former, by letting the doorway be the very first instruction. By contrast, the reverse implication in general does not hold: imagine a \lock\ protocol with a doorway made up by at least its first two instructions, run by a process that sleeps for a very long time between any of them; then, there may be no upper bound on the possible bypasses that happen during this period (that are ignored by post-doorway bounded bypass, but counted by bounded bypass). 

To make this reasoning more concrete, and to show that our definition has a relevance for what has already appeared in the literature, consider Anderson's protocol \cite{And93}, whose code is reported in Algorithm~\ref{Anderson}, to ease the following discussion.

\begin{algorithm}[h]
  \caption{Anderson's Algorithm \cite{And93}. Line 1 is the preliminary initialization; lines 3 and 14 are the way in which $p_0$ and $p_1$ invoke \lock/\unlock.}
  \label{Anderson}
  \begin{algorithmic}[1]

   \State{Initialize P[0], P[1], Q[0], Q[1], T[0] and T[1] to 1}

   \State{\ }

   \State{{\lock}(0) :=    
   \hspace*{5cm}{\lock}(1) := }
   \State\hspace{\algorithmicindent}{P[0] $\leftarrow$ 0}
   \hspace{5.4cm}{P[1] $\leftarrow$ 0}
 
 \State\hspace{\algorithmicindent}{Q[0] $\leftarrow$ 0}
   \hspace{5.38cm}{Q[1] $\leftarrow$ 0}

 \State\hspace{\algorithmicindent}{$x$ $\leftarrow$ T[1]}
   \hspace{5.4cm}{$x$ $\leftarrow \neg$T[0]}
 
 \State\hspace{\algorithmicindent}{T[0] $\leftarrow x$ }
   \hspace{5.3cm}{T[1] $\leftarrow x$ }
    \State\hspace{\algorithmicindent}{{\bf if} $x$ = 1}
\hspace*{5.6cm}{{\bf if} $x$ = 1}
   \State\hspace{\algorithmicindent}{{\bf then\ } P[0] $\leftarrow$ 1}
   \hspace*{4.3cm}{{\bf then\ } Q[1] $\leftarrow$ 1}
   \State\hspace{\algorithmicindent}\hspace{\algorithmicindent}\hspace{\algorithmicindent}{{\bf await} P[1] = 1}\hspace*{4.55cm}{{\bf await} P[0] = 1}
    \State\hspace{\algorithmicindent}{{\bf else\ }\ \ \,Q[0] $\leftarrow$ 1}\hspace*{4.4cm}{{\bf else\ }\ \ \,P[1] $\leftarrow$ 1}
    \State\hspace{\algorithmicindent}\hspace{\algorithmicindent}\hspace{\algorithmicindent}{{\bf await} Q[1] = 1}\hspace*{4.55cm}{{\bf await} Q[0] = 1}
      
   \State{\ }

   \State{{\unlock}(0) := 
     \hspace*{4.8cm}{\unlock}(1) := }
    \State\hspace{\algorithmicindent}{P[0] $\leftarrow$ 1}
   \hspace{5.4cm}{P[1] $\leftarrow$ 1}
 
 \State\hspace{\algorithmicindent}{Q[0] $\leftarrow$ 1}
   \hspace{5.38cm}{Q[1] $\leftarrow$ 1}

\end{algorithmic}
\end{algorithm}  

\begin{remark}
\label{rem:AndersonPDBB}
The protocol is designed for two processes only, $p_0$ and $p_1$, that have different \lock\ protocols. First of all, we want to show that Anderson's algorithm does {\em not} satisfy \df{Bounded bypass}; to see this, consider the following execution:
\begin{itemize}
\item $p_0$ executes line 4 and then sleeps for a long time;
\item then, $p_1$ executes lines 4/../7, finds its $x$ to be 0, jumps to lines 11 and 12, enters its critical section, and unlocks;
\item if now $p_1$ invokes \lock\ again, it will be exactly in the very same situation as before (since T[0] has not changed), and so will enter the critical section again; and so on, until $p_0$ wakes up.
\end{itemize}
This is a concrete example of what we informally discussed above, with the doorway made up of more than one instruction and one process that goes to sleep after executing the first of them.
\end{remark}

In contrast, we now show that Anderson's protocol satisfies \df{Doorway bounded bypass}; 
to ease reasoning, we confine ourselves to the model of atomic read and write registers.

\begin{proposition}{DBBAnderson}
If $p_i$ terminates the execution of line 7 at time $t$, then $p_{1-i}$ can bypass $p_i$ at most once after time $t$.
\end{proposition}
\begin{proof}
We spell out the case for $i = 0$ and for T[0] = 1 at time $t$ (the other cases, viz. for T[0] = 0 at time $t$ and for $i=1$, are handled similarly). At time $t$, $p_0$ has just finished setting $T[0]$ to $x$, which is $1$ in this case. It then proceeds to find $1$ in $x$ at line 8; so, it executes line 9 and arrives at line 10. Here there are two possible cases: either it finds 1 in P[1] and so it enters the critical section without bypasses, or it blocks at line 10, since P[1] is 0. In this latter case, $p_1$ has invoked \lock\ and it has not yet executed line 11 nor line 15. Let us distinguish two cases now:
\begin{itemize}
\item if $p_1$ eventually arrives at line 11, it sets P[1] to 1 and suspends at line 12; this unblocks $p_0$, that can enter the critical section (without bypasses).
\item if not, this means that $p_1$ will eventually take (or has already taken before time $t$) the {\bf then} branch, executes lines 9 and 10, enters its critical section and then unlocks. If $p_0$ is quick enough to react to the execution of line 15 done by $p_1$ in its \unlock, then $p_0$ enters its critical section (after one bypass). Otherwise, $p_1$ was quick enough to reinvoke \lock\ and execute its first line. However, since it finds 1 in T[0], it sets its $x$ and T[1] to 0, and so takes the {\bf else} branch. This leads it to execute line 11 and to block at line 12; as we saw before, this is what allows $p_0$ to enter its critical section (after one bypass now).\qedhere
\end{itemize}
\end{proof}

\noindent
To ease the previous proof, we considered the first four lines of the \lock\ as the doorway. However, we believe that already after the execution of the first two lines the number of bypasses is bounded. However, the proof gets more complicated and the bound might be greater; so, there is no reason for investigating this further.

\begin{remark}
\label{rem:DekkerPDBB}
To conclude the discussion on \df{Doorway bounded bypass}, we now show that post-doorway bounded bypass strengthens starvation-freedom. By definition, every protocol that satisfies the former also satisfies the latter, whereas Dekker's protocol is an example that satisfies the latter but not the former. To see this, we first observe that the only admissible doorway is formed by the first line of the \lock. Indeed, the following execution shows that it is possible to execute the next instructions an unbounded number of times: 
\begin{itemize}
\item let TURN be 0 at the outset;
\item then, $p_0$ and $p_1$ both simultaneously do lines 4 and 5 of the \lock;
\item $p_1$ then sleeps for an arbitrary long time before proceeding to line 6;
\item $p_0$ proceeds to line 6, finds TURN = 0, and jumps back to line 5;
\item $p_0$ again finds FLAG[1] = 1 and TURN = 0; so, it executes lines 5 and 6 an unbounded number of times (until $p_1$ wakes up).
\end{itemize}
Hence, the only possible doorway is line 4 alone; this means that post-doorway bounded bypass for this protocol coincides with bounded bypass as defined in \df{Bounded bypass}. As we discussed above, Dekker's algorithm does {\em not} satisfy \df{Bounded bypass}, so neither \df{Doorway bounded bypass}.
\end{remark}

\subsection{Intermittent Bounded Bypass}\label{sec:ibb}\mbox{}

\bigskip
\noindent
A second alternative to \df{Bounded bypass} is to allow arbitrary bypasses not only during the execution of the very first instruction of the \lock, but during any finite set of (write) actions. 
However, these actions do not cover the period of time between the invocation of the \lock\ and the starting of its first instruction (something that, in contrast, is included in the ``non-counting'' period of Definitions \ref{df:Bounded bypass} and \ref{df:Doorway bounded bypass}). Hence, to define the following notion, we say that a function (in particular, \lock) starts when its first instruction starts.
This leads to the following notion of bounded bypass, that we call {\em intermittent} because such ‘‘bypassing moments'' can arbitrarily appear during the execution of the \lock.

\begin{definition}{intermittent}
A {\MUTEX} protocol satisfies an {\em intermittent bounded-bypass of $f(n)$ within $h(n)$ interrupting assignments} if a set $S$ of assignments to shared registers can be defined, such that each time a process $p_i$ starts  its \lock, 
it will reach its critical section after a period $\pi$ such that
  \begin{itemize}
  \item at most $h(n)$ assignments in $S$ are made within $\pi$; and
  \item $p_i$ enters its critical section losing at most $f(n)$ competitions with the other processes in those parts of $\pi$ that are outside the execution of assignments in $S$.
  \end{itemize}
\end{definition}
\reversemarginpar

The writes in $S$ are referred to as ``interrupting assignments''.
Intuitively, the intermittent bounded-bypass property weakens the bounded-bypass one by allowing a (potentially) arbitrary number of bypasses to happen during the execution of some selected assignments; indeed, bounded bypass, as in \df{Bounded bypass}, is a form of intermittent bounded bypass with one interrupting assignment, namely the first instruction of the \lock\ (that should be a write to the shared memory).
Unlike bounded bypass, where the only interrupting assignment is an assignment by the relevant process itself, intermittent bounded bypass allows \emph{any} assignment to be defined as interrupting, including those performed by other processes.
Pictorially, we can represent the intermittent scenario as follows:

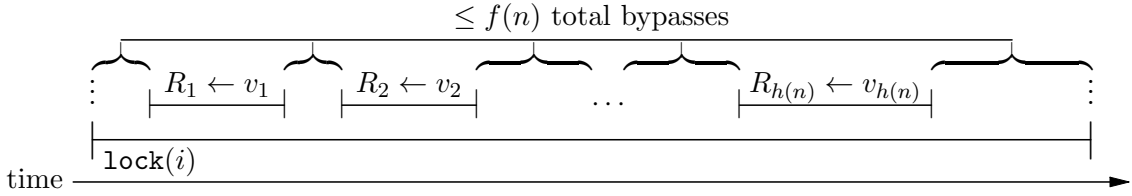
\expandafter\ifx\csname graph\endcsname\relax
   \csname newbox\expandafter\endcsname\csname graph\endcsname
\fi
\ifx\graphtemp\undefined
  \csname newdimen\endcsname\graphtemp
\fi
\expandafter\setbox\csname graph\endcsname
 =\vtop{\vskip 0pt\hbox{%
    \graphtemp=.5ex
    \advance\graphtemp by 0.850in
    \rlap{\kern 0.000in\lower\graphtemp\hbox to 0pt{\hss time\hss}}%
\pdfliteral{
q [] 0 d 1 J 1 j
0.576 w
0.072 w
q 0 g
403.2 -59.4 m
410.4 -61.2 l
403.2 -63 l
403.2 -59.4 l
B Q
0.576 w
14.4 -61.2 m
403.2 -61.2 l
S
Q
}%
    \graphtemp=.5ex
    \advance\graphtemp by 0.650in
    \rlap{\kern 0.300in\lower\graphtemp\hbox to 0pt{\hss \Large $|$\hss}}%
\pdfliteral{
q [] 0 d 1 J 1 j
0.576 w
21.6 -45.36 m
396 -45.36 l
S
Q
}%
    \graphtemp=.5ex
    \advance\graphtemp by 0.750in
    \rlap{\kern 0.600in\lower\graphtemp\hbox to 0pt{\hss $\lock(i)$\hss}}%
    \graphtemp=.5ex
    \advance\graphtemp by 0.650in
    \rlap{\kern 5.500in\lower\graphtemp\hbox to 0pt{\hss \Large $|$\hss}}%
    \graphtemp=.5ex
    \advance\graphtemp by 0.410in
    \rlap{\kern 0.300in\lower\graphtemp\hbox to 0pt{\hss $\vdots$\hss}}%
    \graphtemp=.5ex
    \advance\graphtemp by 0.450in
    \rlap{\kern 0.600in\lower\graphtemp\hbox to 0pt{\hss $|$\hss}}%
    \graphtemp=.5ex
    \advance\graphtemp by 0.350in
    \rlap{\kern 0.950in\lower\graphtemp\hbox to 0pt{\hss $R_1 \leftarrow v_1$\hss}}%
\pdfliteral{
q [] 0 d 1 J 1 j
0.576 w
43.2 -32.4 m
93.6 -32.4 l
S
Q
}%
    \graphtemp=.5ex
    \advance\graphtemp by 0.450in
    \rlap{\kern 1.300in\lower\graphtemp\hbox to 0pt{\hss $|$\hss}}%
    \graphtemp=.5ex
    \advance\graphtemp by 0.450in
    \rlap{\kern 1.600in\lower\graphtemp\hbox to 0pt{\hss $|$\hss}}%
    \graphtemp=.5ex
    \advance\graphtemp by 0.350in
    \rlap{\kern 1.950in\lower\graphtemp\hbox to 0pt{\hss $R_2 \leftarrow v_2$\hss}}%
\pdfliteral{
q [] 0 d 1 J 1 j
0.576 w
115.2 -32.4 m
165.6 -32.4 l
S
Q
}%
    \graphtemp=.5ex
    \advance\graphtemp by 0.450in
    \rlap{\kern 2.300in\lower\graphtemp\hbox to 0pt{\hss $|$\hss}}%
    \graphtemp=.5ex
    \advance\graphtemp by 0.450in
    \rlap{\kern 3.000in\lower\graphtemp\hbox to 0pt{\hss $\cdots$\hss}}%
    \graphtemp=.5ex
    \advance\graphtemp by 0.450in
    \rlap{\kern 3.670in\lower\graphtemp\hbox to 0pt{\hss $|$\hss}}%
    \graphtemp=.5ex
    \advance\graphtemp by 0.350in
    \rlap{\kern 4.170in\lower\graphtemp\hbox to 0pt{\hss $R_{h(n)} \leftarrow v_{h(n)}$\hss}}%
\pdfliteral{
q [] 0 d 1 J 1 j
0.576 w
264.24 -32.4 m
336.24 -32.4 l
S
Q
}%
    \graphtemp=.5ex
    \advance\graphtemp by 0.450in
    \rlap{\kern 4.670in\lower\graphtemp\hbox to 0pt{\hss $|$\hss}}%
    \graphtemp=.5ex
    \advance\graphtemp by 0.410in
    \rlap{\kern 5.500in\lower\graphtemp\hbox to 0pt{\hss $\vdots$\hss}}%
    \graphtemp=.5ex
    \advance\graphtemp by 0.250in
    \rlap{\kern 0.450in\lower\graphtemp\hbox to 0pt{\hss $\overbrace{\rule{.3in}{0pt}}$\hss}}%
    \graphtemp=.5ex
    \advance\graphtemp by 0.150in
    \rlap{\kern 0.449in\lower\graphtemp\hbox to 0pt{\hss $\scriptscriptstyle |$\hss}}%
    \graphtemp=.5ex
    \advance\graphtemp by 0.250in
    \rlap{\kern 1.450in\lower\graphtemp\hbox to 0pt{\hss $\overbrace{\rule{.3in}{0pt}}$\hss}}%
    \graphtemp=.5ex
    \advance\graphtemp by 0.150in
    \rlap{\kern 1.449in\lower\graphtemp\hbox to 0pt{\hss $\scriptscriptstyle |$\hss}}%
    \graphtemp=.5ex
    \advance\graphtemp by 0.250in
    \rlap{\kern 2.600in\lower\graphtemp\hbox to 0pt{\hss $\overbrace{\rule{.6in}{0pt}}$\hss}}%
    \graphtemp=.5ex
    \advance\graphtemp by 0.150in
    \rlap{\kern 2.599in\lower\graphtemp\hbox to 0pt{\hss $\scriptscriptstyle |$\hss}}%
    \graphtemp=.5ex
    \advance\graphtemp by 0.250in
    \rlap{\kern 3.370in\lower\graphtemp\hbox to 0pt{\hss $\overbrace{\rule{.6in}{0pt}}$\hss}}%
    \graphtemp=.5ex
    \advance\graphtemp by 0.150in
    \rlap{\kern 3.369in\lower\graphtemp\hbox to 0pt{\hss $\scriptscriptstyle |$\hss}}%
    \graphtemp=.5ex
    \advance\graphtemp by 0.250in
    \rlap{\kern 5.090in\lower\graphtemp\hbox to 0pt{\hss $\overbrace{\rule{.835in}{0pt}}$\hss}}%
    \graphtemp=.5ex
    \advance\graphtemp by 0.150in
    \rlap{\kern 5.089in\lower\graphtemp\hbox to 0pt{\hss $\scriptscriptstyle |$\hss}}%
\pdfliteral{
q [] 0 d 1 J 1 j
0.576 w
32.4 -7.92 m
366.48 -7.92 l
S
Q
}%
    \graphtemp=.5ex
    \advance\graphtemp by 0.000in
    \rlap{\kern 2.900in\lower\graphtemp\hbox to 0pt{\hss $\leq f(n)$ total bypasses\hss}}%
    \hbox{\vrule depth0.850in width0pt height 0pt}%
    \kern 5.700in
  }%
}%

\centerline{\box\graph}
\vspace{2ex}

\noindent
Similar to the bounded bypass case, also here the first instruction of a \lock\ must be a write to the shared memory; moreover, in this setting, it always counts as interrupting.
We remark that, in reality, each interrupting assignment will take a very short time, in which it is in practice quite unlikely that many bypasses happen.

Clearly, intermittent bounded bypass implies starvation-freedom. 
The reverse implication does {\em not} hold: this is proved by the  execution of Dekker's protocol given in Remark~\ref{rem:DekkerBB}. 
In the next section, we shall present a concrete protocol that satisfies this notion of bounded bypass (in the setting of safe and regular registers), but not the notions of \df{Bounded bypass} and \df{Doorway bounded bypass}. Intermittent bounded bypass is not a generalization of post-doorway bounded bypass: the latter includes the period between the instructions of the doorway, and also allows the doorway to contain reads; the former only includes the duration of each interrupting assignment.

\subsection{A Hierarchy of Liveness Properties}\mbox{}

To conclude this section, we present a hierarchy of liveness properties that includes the notions of bounded bypass devised so far and that relates them to starvation- and deadlock-freedom. This generalizes the chain of implications intended by Raynal in \cite[pg.~11]{Raynal}.

\begin{theorem}{thm:hierarchy}
It holds that
\[
\xymatrix@C=20pt@R=0pt{
& \mbox{post-doorway b.\ b.} \ar@{=>}[dr] &
\\
\mbox{bounded bypass} \ar@{=>}[dr] \ar@{=>}[ur] && \mbox{starvation freedom} \ar@{=>}[r] & \mbox{deadlock freedom}
\\
& \mbox{\ intermittent b.\ b.\ } \ar@{=>}[ur] & \\
}
\]
Moreover, all missing implications do not hold.
\end{theorem}
\begin{proof}
All implications of the Theorem hold by definition; we are left to prove that no other implication holds.
(To better visualize a sample protocol that satisfies/does not satisfy a property, see Table~\ref{tab:summary}.)
\begin{itemize}
\item It is well-known that deadlock-freedom does not imply starvation-freedom (see, e.g., \cite{Raynal}). 
\item {\em Starvation-freedom does not imply post-doorway/intermittent bounded bypass:} Dekker’s algorithm satisfies starvation-freedom; however, it does not satisfy any form of bounded
bypass (see Remark~\ref{rem:DekkerPDBB} for post-doorway bounded bypass and the execution contained in Remark~\ref{rem:DekkerBB} for intermittent bounded bypass).
\item {\em Post-doorway bounded bypass does not imply bounded bypass:} Anderson’s algorithm does not satisfy \df{Bounded bypass} (Remark~\ref{rem:AndersonPDBB}), but satisfies \df{Doorway bounded bypass} (Proposition \ref{pr:DBBAnderson}).

\item {\em Intermittent bounded bypass does not imply bounded bypass:} In the next section, we will present Algorithm~\ref{alg}, which in the setting of safe registers does not satisfy bounded bypass (Proposition \ref{pr:No Bounded Bypass}) but does satisfy intermittent bounded bypass (Theorem~\ref{thm:bib}).

\item {\em Intermittent bounded bypass and post-doorway bounded bypass are logically unrelated:} If a protocol satisfies intermittent bounded bypass, the $h(n)$ write actions may not be included in the execution of a piece of code that we consider an acceptable doorway (actually, they may also be actions belonging to the execution of the \lock\ protocol by different processes
 -- see, e.g., the proof of Lemma~\ref{lem:bounded on interrupts} later on); hence, it needs not satisfy post-doorway bounded bypass. In fact, the forthcoming Algorithm~\ref{alg} does not satisfy post-doorway bounded bypass (this can be obtained by a reasoning similar to Proposition \ref{pr:No Bounded Bypass}).

For the converse, consider again Anderson's protocol; it is possible that, by stopping a process after the execution of the first instruction, it can be bypassed an unbounded number of times before it executes its second instruction. This is enough to show that this protocol fails to satisfy intermittent bounded bypass.\qedhere
\end{itemize}
\end{proof}

\begin{table}[t]
\centering
\begin{tabular}{l|c||clc}
\textbf{Protocol} & \multicolumn{1}{l||}{\textbf{Memory model}} & \multicolumn{1}{l|}{\textbf{B.B.}} & \multicolumn{1}{l|}{\textbf{Post-doorway B.B.}} & \multicolumn{1}{l}{\textbf{Intermittent B.B.}} \\ \hline \hline
\textit{Dekker} & Atomic & \multicolumn{3}{c}{X} \rule{0pt}{12pt}\\ \hline
\textit{Anderson} & Atomic & \multicolumn{1}{c|}{X} & \multicolumn{1}{c|}{$f(n) = 1$} & X \rule{0pt}{12pt}\\ \hline
\textit{Bar-David} & Atomic & \multicolumn{2}{c|}{$f(n) = n(n{-}1){-}1$} & \begin{tabular}[c]{@{}c@{}}$f(n) = n(n{-}1){-} 1$\\ $h(n) = 1$\end{tabular} \rule{0pt}{20pt}\rule[-12pt]{0pt}{0pt}\\ \hline
\rule{0pt}{20pt}\textit{Bar-David} & Safe/Regular & \multicolumn{2}{c|}{X} & \begin{tabular}[c]{@{}c@{}}$f(n) = n^2{-}2$\\ $h(n) = n$\end{tabular}
\end{tabular}
\caption{Summary of the Bounded Bypass (B.B.) properties satisfied by Algorithms \ref{Dekker}, \ref{Anderson} and \ref{alg} (to be presented in the next section), according to the memory models adopted (safe/regular/atomic registers). Here, ‘X' means that the property is not satisfied and an arithmetical expression means that the property is satisfied with that bound.}
\label{tab:summary}
\end{table}

Of course, we can combine the two notions of bounded bypass into a fourth property, called {\em post-doorway intermittent bounded bypass}, that ignores the bypasses that may happen during an initial portion of the \lock\ {\em and} during a bounded number of writes that may happen in successive moments. This property still strictly implies starvation-freedom (again, Dekker's algorithm can be used to show this).
However, we are not aware of any documented protocol that merely satisfies this fourth form of bounded bypass.

\section{On Enhancing the Liveness Property of any Deadlock-free Algorithm}
\label{sec:alg}

As we have just discussed, the different forms of bounded bypass are among the strongest liveness properties for {\MUTEX} protocols that occur in the literature. We now show that any deadlock-free {\MUTEX} protocol, call it DLF, can be exploited to generate a new {\MUTEX} protocol that satisfies (intermittent) bounded bypass.
 This idea was first presented in \cite[Exercise 2.34]{T}, and declared to be devised by Yoah Bar-David in 1998; it was then formally developed in \cite[Sect.\ 2.2.2]{Raynal}, where it was proved to promote any deadlock-free protocol to a starvation-free one, in the setting of atomic registers.
The proposed protocol is reported as Algorithm \ref{alg}.

This protocol maintains a Boolean variable FLAG[$i$] for each process $p_i$ (for $i \in \{0,\dots,\linebreak[1] n{-}1\}$),
as well as a variable TURN that can take any value in $\{0,\dots,n{-}1\}$. The variable FLAG[$i$] tells whether $p_i$ wants to enter its critical section; it is set to $1$ at the beginning of the {\lock} function of $p_i$, and to $0$ at the beginning of its \unlock. The variable TURN tells which is the most privileged process to enter the critical section; its value cycles through all process IDs, and is incremented only in the \unlock, whenever $p_{\text{TURN}}$ has not signaled its intention to enter the critical section. For the rest, the protocol is extremely simple: a process $p_i$ that wishes to enter the critical section waits until either TURN becomes $i$, or $p_{\text{TURN}}$ has not signaled intention to enter. After this wait, the process runs the {\lock} of the given protocol DLF and, naturally, the {\unlock} of DLF is run to complete the protocol.

\cite{Raynal,T} work in the setting of atomic registers, for which they prove that Algorithm \ref{alg} promotes DLF to a starvation-free protocol. In the next two sections, we prove that the very same protocol (1) satisfies bounded bypass (with a bound that is quadratic in the number of processes involved) for atomic registers, and (2) promotes DLF to an intermittent bounded bypass protocol also in the (weaker) settings of safe and regular registers.

\begin{algorithm}[t]
  \caption{Enhancing the Liveness of a deadlock-free {\MUTEX} protocol (Bar-David 1998). Line~1 is the preliminary initialization; lines 3 and 10 are the way in which \lock/\unlock\ are invoked by process $p_i$, for $i \in \{0,...,n{-}1\}$. The symbolic names to read/write actions to the shared memory are used in the proofs.}
  \label{alg}
  \begin{algorithmic}[1]

   \State{Initialize FLAG[$i$] to 0, for all $i$, and TURN to any process ID ($\in \{0,\ldots,n{-}1\}$)}

   \State{\ }

   \State{{\lock}($i$) := }
   \State\hspace{\algorithmicindent}{FLAG[$i$] $\leftarrow$ 1} \Comment{W1}
   \State\hspace{\algorithmicindent}{\bf repeat}
   \State\hspace{\algorithmicindent}\hspace{\algorithmicindent}{$tmp \gets $ TURN} \Comment{R1}
   \State\hspace{\algorithmicindent}{{\bf until }($tmp = i\ \vee$ FLAG[$tmp$] = 0}) \Comment{R2}
   \State\hspace{\algorithmicindent}{DLF.\lock($i$)}
   
   \State{\ }

   \State{{\unlock}($i$) := }
   \State\hspace{\algorithmicindent}{FLAG[$i$] $\leftarrow$ 0} \Comment{W2}
   \State\hspace{\algorithmicindent}{$tmp \gets $ TURN} \Comment{R3}
   \State\hspace{\algorithmicindent}{{\bf if }(FLAG[$tmp$] = 0){\ \bf then}} \Comment{R4}
   \State\hspace{\algorithmicindent}\hspace{\algorithmicindent}{TURN $\gets (tmp+1) \mod\ n$} \Comment{W3}
   \State\hspace{\algorithmicindent}{DLF.\unlock($i$)}
   
\end{algorithmic}
\end{algorithm}  

To carry out the proofs, we find it convenient to refer to the read/write actions (of shared registers) as follows:
\begin{itemize}
\item W1/W2/W3 are the write actions happening, respectively, in lines 4, 11 and 14;
\item R1/R2/R3/R4 are the read actions happening, respectively, in lines 6, 7, 12 and 13.
\end{itemize}

\subsection{The Algorithm with Atomic Registers}\mbox{}

\bigskip
\noindent
For atomic registers, \cite{Raynal} proves that Algorithm \ref{alg} promotes any deadlock-free {\MUTEX} protocol to a starvation-free one; here, we push this result further and prove that the resulting protocol enjoys bounded bypass, with a bound that is quadratic in the number of processes.

Since DLF satisfies mutual exclusion---never will multiple processes be in the critical section---so does Algorithm \ref{alg}. We start by showing that Algorithm \ref{alg} is deadlock-free.

\begin{lemma}{dlf}
If at some time $t$ a process is running its \lock, then afterwards a process will enter its critical section.
\end{lemma}
\begin{proof}
  Since DLF enjoys deadlock-freedom, it suffices to prove that at least one process either is running DLF.\lock\ at time $t$, or invokes DLF.\lock\ after time $t$. Towards a contradiction, assume that at time $t$ no process is running DLF.\lock\, and no process ever will. 
  Let $t' (\geq t)$ be the time in which any $p_j$ that has already completed DLF.\lock\ at time $t$ reaches line~15 (if no such $p_j$ exists, then $t' {=} t$).
Let TURN\,=\,$k$ at time $t'$; then TURN will remain $k$ forever.

Let $p_i$ be running its {\lock} at time $t$. By our assumption, it is not running DLF.{\lock} and will never invoke it; this is possible only if $p_i$ remains in the {\bf repeat} loop forever and, in particular, after time $t'$. For this to happen, $i \neq k$ and, after time $t'$, $p_i$ reads TURN = $k$ in R1 and FLAG[$k$] = 1 in R2. But in that case $p_k$ is also running its \lock\ at time $t'$; it finds TURN = $k$ in R1, exits the {\bf repeat} loop, and thus eventually invokes DLF.{\lock}. This yields the required contradiction.
\end{proof}

Assume that $n>1$.
Since DLF is assumed to be a correct {\MUTEX} protocol, at most one process can be executing the first segment of \unlock, lines 11--14. Namely, two processes doing this are, from the perspective of DLF, indistinguishable from two processes in the critical section, and that has to be excluded. As a consequence, the variable TURN cannot change between R3 and R4, so that R4 actually checks whether FLAG[TURN] $=0$.

\begin{definition}{stably}
We say that the variable TURN is \emph{stably} $i$ \emph{at time $t$}, notation TURN $\equiv_t i$, whenever at time $t$ TURN $= i$ and no process $j$ is running the part of \unlock$(j)$ prior to line 15.
\end{definition}

\begin{lemma}{lOne}
If TURN $\equiv_t i$ and FLAG[$i$] = 1 at time $t$, then no process different from $p_i$ can progress from line 3 to line 8 between time $t$ and the moment in which $p_i$ executes W2.
\end{lemma}
\begin{proof}
  Under the given hypotheses, FLAG[$i$] remains 1 and the value of TURN does not change from time $t$ until $p_i$ executes W2.
  Until then, no process other than $i$ can progress from line 3 to line 8, since every other $p_j$ gets blocked in its {\bf repeat}.
\end{proof}

\begin{corollary}{cOne}
  If TURN $\equiv_t i$ and FLAG[$i$] = 1 at some time $t$ when $p_i$ is running its \lock, then $p_i$ enters the critical section after losing at most one competition with each of the other processes. 
\end{corollary}

\begin{proof}
Since $p_i$ is running its \lock\ at time $t$ and TURN remains $i$, $p_i$ will eventually invoke (or has already invoked) DLF.\lock.
Because of Lemma \ref{lem:lOne}, a process $j\neq i$ can enter its critical section after time $t$ and before $p_i$ does only if $p_j$ invoked its \lock\ before time $t$. No $p_j$ can do this twice. Consequently, $p_i$ can lose at most one competition with each of the other processes before entering its critical section.
\end{proof}

\begin{lemma}{lTwo}
If $p_i$ runs {\lock} at time $t$,  then afterward the value of TURN will be incremented, and this happens after $p_i$ loses at most $n{-}1$ competitions.
\end{lemma}
\begin{proof}
  Let $p_j$ be the first process that finishes $\lock$ after time $t$ (such a process exists because of Lemma~\ref{lem:dlf}). Right after doing this, $p_j$ will enter its critical section and subsequently unlock.
If within this \unlock\ $p_j$ finds FLAG[TURN] = 0, then TURN is increased. Otherwise, TURN $\neq j$ and, by using Lemma \ref{lem:lOne} with the time in which $p_j$ reaches line 15, $p_j$ can not win again, and no other process can win twice until $p_{\rm TURN}$ wins the competition.
When $p_{\rm TURN}$ wins, which is guaranteed to happen by Corollary~\ref{cor:cOne}, it increases TURN in its {\unlock} phase.
Hence $p_i$ loses at most $n{-}1$ competitions before the value of TURN is incremented.
\end{proof}

\begin{theorem}{bbp}{\bf (Bounded bypass with atomic registers)}
When a process completes W1, then it enters the critical section after losing at most $n(n{-}1){-}1$ competitions.
\end{theorem}
\begin{proof}
Say that $p_i$ completes W1 at time $t$; the worst case is when TURN = $(i{+}1) \mod n$ at time $t$.\footnote{It could be considered worse if TURN = $i$ at time $t$, but another process $j$ in its {\unlock} sets TURN to $(i{+}1)\! \mod n$ in W3, say at time $t'>t$. In that case $p_j$ must have read FLAG[$i$] = 0 in R4, which must have occurred prior to time $t$. So no process has won the competition between $t$ and $t'$, and we can simply start our reasoning from $t'$.} By Lemma \ref{lem:lTwo}, $p_i$ may lose at most $n{-}1$ competitions before the value of TURN is incremented and so it  may lose at most $(n{-}1)^2$ competitions before TURN reaches the value $i$. In case $p_i$ completes its {\lock} before TURN reaches the value $i$ we are done, so assume that $p_i$ is still running {\lock} when TURN reaches the value $i$.
Right after this occurs, we reach a time $t'$ when TURN $\equiv_{t'} i$, \mbox{FLAG[$i$] = 1}, and $p_i$ is still running its \lock. After $t'$, by Corollary~\ref{cor:cOne}, $p_i$ may lose at most one competition with each of the other processes before entering its critical section. Moreover, by Lemma~\ref{lem:lOne}, after $t'$ it cannot lose another competition with the process that incremented the value of TURN to $i$. Thus, it loses at most $n{-}2$ competitions after time $t'$. This brings the total losses to at most $(n{-}1)^2 + (n{-}2) = n(n{-}1){-}1$.
\end{proof}

Note that the bound given in Theorem~\ref{thm:bbp} is tight, in the sense that it can be achieved in some execution of Algorithm~\ref{alg}. To see this, let TURN be initialized to 1 and let us prove that $p_0$ can be bypassed exactly $n(n{-}1){-}1$ times. Indeed:
\begin{enumerate}
  \leftmargini 1in
\item $p_0$ executes W1.
\item Processes from 2 to $n{-}1$ execute the \lock\ and assemble at the start of line 8 (they pass line 7 because FLAG[TURN]=0).
\item $p_1$ executes W1.
\item One by one, processes from 2 to $n{-}1$ execute DLF.\lock\ without contention, enter their critical sections and execute their \unlock, 
without incrementing TURN (since FLAG[TURN]=1);
this step yields $n{-}2$ bypasses.
\item $p_1$ executes the remainder of its \lock, its critical section and its \unlock\ (where TURN is incremented, since FLAG[TURN]=0);
this yields another bypass.
\item We repeat steps (2)--(5) above $n{-}3$ times, with $p_{\rm TURN}$ (for TURN = 2, 3, $\ldots, n{-}2$) in the role of $p_1$ and processes $\{1,...,n{-}1\} \setminus \{p_{\rm TURN}\}$ in the roles of $p_2,\ldots,p_{n-1}$.
\item We repeat steps (2)--(4) above with $p_{n-1}$ in the role of $p_1$ and processes $\{1,...,n{-}2\}$ in the roles of $p_2,\ldots,p_{n-1}$.
\item $p_{n-1}$ executes the remainder of its \lock, its critical section and line 11 of its \unlock. 
\item Now, processes from 1 to $n{-}2$ execute the \lock\ once more and assemble at the start of line 8
(again, they pass line 7 because FLAG(TURN)=0).
\item $p_{n-1}$ finishes its \unlock, by putting TURN to 0.
\item One by one, processes from 1 to $n{-}2$ execute DLF.\lock\ without contention, their critical sections, and the \unlock\ (here, TURN is not incremented, since FLAG[TURN]=1).
\item $p_0$ finally executes the remainder of its \lock\ and enters its critical section.
\end{enumerate}
Steps (1) to (8) bring $(n{-}1)(n{-}1)$ bypasses, whereas steps (9) to (11) bring another $n{-}2$ bypasses. So, the total number of bypasses is $(n{-}1)(n{-}1) + (n{-}2) = n(n{-}1){-}1$.

\subsection{The Algorithm with Safe and Regular Registers}\label{sec:algSafeRegs}
\mbox{}

\bigskip
\noindent
Let us now consider the behavior of the algorithm in the weakest possible memory model, namely the one with safe registers.
First of all, notice that
the elements of FLAG are SWMR safe registers, since FLAG[$i$] can only be written by $p_i$ through W1 and W2. Moreover, the write W3 occurs in a portion of the code (lines 11--14) that, from the point of view of DLF, is treated equivalent to a critical section; hence, there are no overlapping writes.

\begin{proposition}{No Bounded Bypass} 
For Algorithm \ref{alg}, no bound on the number of bypasses can be guaranteed with safe registers. 
\end{proposition}

\begin{proof}
Consider a scenario with three processes $p_0$, $p_1$ and $p_2$. Assume that TURN = 0 and that $p_0$ is staying within its W1 for a very long time. During this time, assume that $p_1$ always reads FLAG[0] = 1 in its R2 (and so never reaches its line 8), whereas $p_2$ always reads FLAG[0] = 0 in its R2, proceeds to line 8, enters the critical section and reads FLAG[0] = 1 in its R4 (and so does not increase TURN). This can be repeated an unbounded number of times, thereby breaking any bound on the number of losses $p_1$ may endure before entering its critical section.
\end{proof}

\begin{remark}\label{rem:regbb}
The scenario above applies to regular registers as well, so also for regular registers Algorithm \ref{alg} does not satisfy bounded bypass.
\end{remark}

We will now show, when $n{>}2$, that with safe registers Algorithm~\ref{alg}
satisfies an intermittent bounded bypass of $n^2{-}2$ within $n$ interrupting assignments. For this algorithm, the set $S$ of interrupting assignments is made up by:
\begin{itemize}
\item the W1 operation of the process $p_i$ considered, and 
\item those W1 operations of any other process $p_j$ involved in the protocol for which there exists a moment during its execution when 
\begin{enumerate}
\item[(a)] the W1 operation of the process $p_i$ has been completed, 
\item[(b)] $\lock(i)$ has not yet been completed, and 
\item[(c)] the register TURN is not being written and its value is $j$.
\end{enumerate}
\end{itemize}
\noindent
By comparison, the previous section showed that, with atomic registers, Algorithm \ref{alg} 
satisfies an intermittent bounded bypass of $n(n{-}1){-}1$  within $1$ interrupting assignment, namely the W1 executed by the process considered. That assignment is the one referred to in the last part of \df{Bounded bypass}, which we no longer need in \df{intermittent}.

First note that Lemma \ref{lem:dlf} still holds for safe registers.
The reason is that its proof nowhere depends on the reading of a shared register that could be written by another process at the same time.

We now write FLAG$[i] \doteq_t 1$ if, at time $t$, FLAG$[i] = 1$ and $p_i$ is not executing W1 or W2. If in Lemma \ref{lem:lOne} and Corollary \ref{cor:cOne} we replace FLAG$[i] = 1$ with FLAG$[i] \doteq_t 1$, then these results hold also for safe registers, for the same reason as above.

\begin{definition}{stable time}
Given any time $t$, let $\widehat t$ be the first time $\geq t$ in which 
every process that has finished its \lock\ at time $t$ has reached line 15.
\end{definition}

\begin{lemma}{stable}
For every $t$, no process can finish its \lock\ between $t$ and $\widehat t$.
\end{lemma}
\begin{proof}
If at time $t$ all processes that  finished their \lock\ have also already reached line 15, then $\widehat t = t$. Otherwise, because of mutual exclusion of DLF.\lock, there is one process that at time $t$ finished its \lock\ but has not yet reached line 15; so, still because of mutual exclusion, no other process can finish its \lock\ between time $t$ and $\widehat t$.
\end{proof}

\begin{lemma}{lTwo safe}
If at time $t$ process $p_i$ is running {\lock} and has already completed its W1 operation, then afterward the value of TURN will be incremented, and this happens after $p_i$ loses at most $n$ competitions, excluding the competitions lost during the execution of an interrupting W1 by $p_{\rm TURN}$.
\end{lemma}
\begin{proof}
  In case TURN gets incremented after time $t$ and before any process finishes {\lock}, there is nothing to prove. Thus, assume that this is not the case.
  Let $p_j$ be the first process that finishes $\lock$ after time $t$ (such a process exists because of Lemma~\ref{lem:dlf}). Right after doing this, $p_j$ will enter its critical section and subsequently unlock. Let $t'$ be the time when $p_j$ completes the execution of R4 in its \unlock.
If within this \unlock\ $p_j$ finds FLAG[TURN] = 0, then TURN is incremented and the lemma is proved. 

Otherwise, TURN $\neq j$ and at time $t'$ process $p_{\rm TURN}$ either is executing its W1 or has already done so. Let $u_{\rm TURN}$ and $u'_{\rm TURN}$ be the start and termination times of this execution of W1 by $p_{\rm TURN}$, and let $t_{\rm TURN} = \max\{t',u_{\rm TURN}\}$ and $t'_{\rm TURN} = \max\{t',u'_{\rm TURN}\}$.
Between $t$ and $t_{\rm TURN}$ $p_i$ loses at most $1$ competition, namely with $p_j$, in case $j \neq i$; indeed, the only way for $p_j$ to find FLAG[TURN] $\neq 0$ is to complete R4 after $u_{\rm TURN}$, so $t_{\rm TURN}=t'$.
After \plat{$\widehat{t'}_{\rm TURN}$}, by Corollary \ref{cor:cOne} $p_{\rm TURN}$ will enter the critical section after losing at most one competition with each of the other processes. When this happens, $p_{\rm TURN}$ increases TURN in its {\unlock} phase. Hence $p_i$ loses at most $n{-}1$ competitions after $t'_{\rm TURN}$ and before the value of TURN is incremented.
Thus, the total number of losses excluding the ones that happen between $t_{\rm TURN}$ and $t'_{\rm TURN}$ is $n$.

In the special case that $j=i$, process $p_i$ suffers no losses at all before TURN gets incremented, and we are done. Otherwise, at time $t'$ the W1 operation of the process $p_i$ has been completed, but $\lock(i)$ has not yet been completed. Moreover, at time $t'$ the value of TURN is not being written. If at that moment the W1 of $p_{\rm TURN}$ is still going on, then that W1 counts as an interrupting assignment.
Moreover, $t'<u'_{\rm TURN}$, and the interval between $t_{\rm TURN}$ and $t'_{\rm TURN}$ is fully contained in the execution of that W1. Otherwise $t'\geq u'_{\rm TURN}$ and the interval between $t_{\rm TURN}$ and $t'_{\rm TURN}$ is empty.
\end{proof}

\begin{lemma}{bounded on interrupts}
\label{bounded on interrupts}
During an execution of $\lock(i)$, at most $n$ interrupting assignments can occur.
\end{lemma}
\begin{proof}
  Recall that for this protocol we defined an interrupting assignment as the W1 operation of $p_i$, and those W1 operations of any process $p_j \neq p_i$ for which there exists a moment during its execution when (a) the W1 operation of the process $p_i$ has been completed, (b) $\lock(i)$ has not yet been completed, and (c) the register TURN is not being written and its value is~$j$. Let $t$ be the time at which $p_{i}$ completes its W1 instruction.  We have to show that, after $t$ and before $\lock(i)$ terminates, at most $n{-}1$ interrupting assignments occur---including interrupting assignments that start before time $t$, or end after $\lock(i)$ terminates.

If, at time $t_i > t$, TURN ever reaches the value $i$, by some process completing W3, it will not be incremented further until $p_i$ has completed its \lock.  If $p_j \neq p_i$ executes an interrupting assignment, then there is a time $t_j > t$ when $p_j$ is executing its W1 operation, $\lock(i)$ has not yet been completed and 
  the register TURN is not being written and its value is~$j$. Hence in $\unlock(j)$, if it ever occurs, if the value of TURN still is $j$, it will be incremented. Thus, each  $p_j \neq p_i$ can execute at most one interrupting assignment.
\end{proof}

\begin{theorem}{bib}{\bf (Intermittent bounded bypass with safe registers)}
Algorithm~\ref{alg} satisfies an intermittent bounded bypass of $n^2{-}2$ within $n$ interrupting assignments.
\end{theorem}
\begin{proof}
Say that $p_i$ invokes \lock$(i)$ at time $t_0$ and completes W1 at time $t > t_0$; by Lemma \ref{lem:stable}, no process will finish its {\lock} between $t$ and~$\widehat t$. The worst case is when TURN = $(i{+}1) \mod n$ at time $\widehat t$. By Lemma \ref{lem:lTwo safe}, $p_i$ may lose at most $n$ competitions within 1 interrupting assignment before the value of TURN is incremented, and so it may lose at most $(n{-}1)n$ competitions within $n{-}1$ interrupting assignments before TURN reaches the value $i$. In case $p_i$ completes its {\lock} before TURN reaches the value $i$ we are done, so assume that $p_i$ is still running {\lock} when TURN reaches the value $i$.
Right after this occurs, we reach a time $t'$ when TURN $\equiv_{t'} i$ and $p_i$ is still running its \lock.

Since $p_i$ has already completed W1 at time $t< t'$, FLAG[$i$] $\doteq_{t'} 1$. Thus, by Corollary~\ref{cor:cOne}, after $t'$ $p_i$ may lose at most one competition with each of the other processes. Moreover, by Lemma~\ref{lem:lOne}, after $t'$ it cannot lose another competition with the process that incremented the value of TURN to $i$. Thus, it loses at most $n{-}2$ competitions after time $t'$. This brings the total losses to at most $(n{-}1)n + (n{-}2) = n^2{-}2$, within $n$ interrupting assignments (including the execution of W1 between $t_0$ and $t$).
\end{proof}

Also in this case, the bound given in Theorem~\ref{thm:bib} is tight, at least when $n{>}2$. To see this, let TURN be initialized to 1 and let us prove that $p_0$ can be bypassed exactly $n^2{-}2$ times. Indeed:
\begin{enumerate}
\item $p_0$ invokes \lock\ and executes its W1.
\item Let $p$ be any process in $\{p_1,\ldots,p_{n-1}\}$ different from $p_{\rm TURN}$ and let it invoke the \lock, complete it (since it reads FLAG[TURN]=0) and execute its critical section.
\item $p_{\rm TURN}$ invokes \lock\ and starts W1.
\item $p$ executes \unlock\ without changing TURN (indeed, it may read 1 from FLAG[TURN]).
\item All processes in $\{p_1,\ldots,p_{n-1}\}$ different from $p_{\rm TURN}$ now invoke the \lock\ and assemble at line 8 (since they can all read 0 from FLAG[TURN]). 
\item $p_{\rm TURN}$ completes W1.
\item All processes in $\{p_1,\ldots,p_{n-1}\}$ different from $p_{\rm TURN}$ one by one execute DLF.\lock\ without contention, enter their critical sections and execute the \unlock\ without changing TURN (indeed, they all read 1 from FLAG[TURN]).
\item $p_{\rm TURN}$ completes the \lock, executes its critical section and the \unlock, where it increments TURN. The bypasses undergone by $p_0$ so far are $n$.
\item Now, repeat steps (2)--(8) again, for TURN = 2, 3, $\ldots,n-2$.
\item Finally, repeat steps (2)--(7) for TURN = $n{-}1$; then, let $p_{n-1}$ complete the \lock, execute its critical section and the first line of the \unlock.
\item All processes in $\{p_1,\ldots,p_{n-2}\}$ now invoke the \lock\ and assemble at line 8 (since they all read 0 from FLAG[TURN]).
\item $p_{n-1}$ then completes the \unlock, by setting TURN to 0.
\item All processes in $\{p_1,\ldots,p_{n-2}\}$ one by one execute DLF.\lock\ without contention, enter their critical section and  execute the \unlock\ without changing TURN.
\item To conclude, $p_0$ enters its critical section.
\end{enumerate}
Steps (2) to (10) bring $n(n{-}1)$ bypasses, whereas steps (11) to (13) bring another $n{-}2$ bypasses. So, the total number of bypasses is $n(n{-}1) + (n{-}2) = n^2 {-}2$. 

Notice that the above argument does not go through when $n=2$, as in step (2) no $p$ with the required properties can be found. In this case, only 1 bypass is possible, with only 1 interrupting assignment (namely, the W1 operation of the bypassed process). This has been confirmed by means of model checking---see \Sec{model checking}.

\begin{remark}\label{rem:regular}
Since regular registers allow only a subset of the possible behaviors of safe registers, trivially Theorem~\ref{thm:bib} also applies to regular registers. Moreover, the execution witnessing that the bound is tight applies easily to regular registers as well, since the only overlaps are on Boolean registers (flags) and they are only written to change the value; safe and regular registers coincide in this case.
\end{remark}

\section{Model Checking}\label{sec:model checking}
\newcommand{\PID}{\ensuremath{\mathbb{P}}}
\newcommand{\pid}{\ensuremath{p}}
\newcommand{\swby}[1]{\ensuremath{\mathit{sw_{#1}}}}
\newcommand{\fwby}[1]{\ensuremath{\mathit{fw_{#1}}}}
\newcommand{\startlock}[1]{\ensuremath{\mathit{sl}_{#1}}}
\newcommand{\finishlock}[1]{\ensuremath{\mathit{fl}_{#1}}}
\newcommand{\acts}{\textit{acts}}
\newcommand{\actsexcept}[1]{\ensuremath{\acts_{\!/\!#1}}}
\newcommand{\false}{\ensuremath{\mathit{false}}}
\newcommand{\flexcept}[1]{\ensuremath{\finishlock{\!/\!#1}}}

In \cite{GLS25} the mutual exclusion and starvation-freedom properties of Algorithms \ref{Dekker}, \ref{Anderson} and \ref{alg} have been automatically verified using the model checker of the mCRL2 toolset~\cite{mCRL2toolset}. 
Notably, Dekker's and Anderson's algorithms are found to satisfy mutual exclusion and starvation-freedom with atomic registers.
In the case of Algorithm~\ref{alg}, Lamport's One-Bit protocol is used as DLF. 
The resulting protocol, verified up to 3 processes, is found to behave as well as Lamport's Three-Bit protocol \cite{LamportMutex2}, as it satisfies starvation-freedom for all the three models considered in this paper (safe/regular/atomic registers).

The various forms of bounded bypass we presented in this paper have not been formalized in \cite{GLS25}, but we do this here. In order to reapply the work of  \cite{GLS25}, all we need to do is to capture these properties as formulae in the version of the modal $\mu$-calculus \cite{Pratt81,Koz83} employed by mCRL2 (that allows recursion variables to be parametrised by higher-order data -- such as numbers, sets and functions), and make minor changes to the models to support these formulae. 
Of note is that mCRL2 models have a single initial state. When verifying algorithms with multiple initial states, these being Dekker's and Bar-David's algorithms due to the different possible initial values of TURN, we check every initial state separately.

In Section~\ref{sec:bypass}, it was stated that bounded bypass can be viewed as a safety property that is required in conjunction with deadlock- and starvation-freedom, or as a liveness property that implies deadlock- and starvation freedom. Throughout this paper, the latter interpretation is used. However, for formalizing the properties it is much simpler to only consider the safety aspect. Deadlock- and starvation-freedom are formalized and verified separately in \cite{GLS25}; thus it is sufficient to only check the safety aspect here.

We now provide a high-level description of the formulae for the different notions of bounded bypass;
their concrete syntax in mCRL2, as well as the mCRL2 models of the algorithms with the different register types, are available online at \href{https://doi.org/10.5281/zenodo.17650539}{Zenodo}~\cite{GGSzenodo}.
We also report on the results of checking these formulae with mCRL2.
To do this, we follow the approach in \cite{GLS25} and run the model checker only for 2 or 3 processes; going beyond this limit would make the checking practically unfeasible due to state-space explosion.

Below, we write $\PID = \{p_0,\dots,p_{n-1}\}$ for the set of $n$ processes participating in a mutual exclusion protocol. We model such a protocol as a labeled transition system that features the actions $\startlock{\pid}$ and $\finishlock{\pid}$, stating that process $\pid \in \PID$ starts and finishes its \lock.
Notice that finishing  \lock\ in this model coincides with entering into the critical section; hence, actions $\finishlock{\pid}$ will be used to count the bypasses.
We write $\finishlock{\it all}$ for the set $\{\finishlock{q} \mid q \in \PID\}$, the finish-lock actions of all processes, and $\flexcept{p}$ for $\{\finishlock{q} \mid q \in \PID \wedge q \neq p\}$, the finish-lock actions of all processes except $p$.
We use $\swby{\pid}$ and $\fwby{\pid}$ for signaling when a process $\pid \in \PID$ starts and finishes a write operation, respectively;
we define $\swby{\it all}$ and $\fwby{\it all}$ analogously to $\finishlock{\it all}$.
Finally, we write $\acts$ for the set of all actions in the model, with $\acts_{\pid}$ being all actions by process $p$ and $\actsexcept{\pid}$ being all actions by processes other than $p$.

\subsection{Bounded Bypass and Post-Doorway Bounded Bypass}\mbox{}\label{sec:modelcheckingbb}

\bigskip
\noindent
A modal $\mu$-calculus formula partially expressing post-doorway bounded bypass is obtained by adapting the bounded bypass formula presented in \cite{hafidi2021fair}. To this aim,
 we add to our model a special action $\textit{doorway}_p$, stating that process $p$ leaves the doorway section of its lock. In fact, this action is implemented as an extra tag on the final action of the doorway.
 
The formula $\textit{pdbb}_p(b)$ defined below intuitively expresses that each time process $p\in\PID$ invokes the \lock, it finishes the \lock\ after losing at most $b$ competitions from the moment in which it performs the action $\textit{doorway}_p$. To find out whether a given protocol satisfies post-doorway bounded bypass, for a given $n$, a given bound $f(n)$, and a given placement of the doorway, we thus have to check the formula $\bigwedge_{\pid\in\PID}\textit{pdbb}_p(f(n))$. This formula expresses post-doorway bounded bypass only partially, because it does not quantify over all valid placements of the doorway in the protocol, and does not check whether executing the doorway requires only a bounded number of actions. So, a thorough check involves checking this formula for versions of the $n$-process protocol for each reasonable placement of the doorway.

Formally, we let
  \[\begin{array}{rcl}\textit{pdbb}_p(b) & := &
     [\acts^*][\startlock{p}][\overline{\{\textit{doorway}_p\}}^*][\textit{doorway}_p]
     \textit{count\_bypasses}_p(b) \\[2ex]
     \textit{count\_bypasses}_p(b) & := &
     \nu X(b).\big([\overline{\finishlock{\it all}}]X(b) \wedge [\flexcept{p}](b>0 \wedge X(b{-}1))\big)
  \end{array}\]
The formula $[\acts^*]\psi$ says that condition $\psi$ must hold after any sequence of actions, and
thus in any reachable state of the protocol. Likewise $[\startlock{p}]\phi$ says that condition
$\phi$ must hold in any state reachable by a transition labeled $\startlock{p}$. So together the
formula $\textit{pdbb}_p(b) = [\acts^*][\startlock{p}]\phi$ says that some condition $\phi$ must
hold after any execution ending with the start of $p$'s \lock. This condition $\phi$, in turn, says
that the formula \plat{$\textit{count\_bypasses}_p(b)$} must hold after any sequence of actions ending with
the first occurrence of the action $\textit{doorway}_p$. Here $\overline{S}$ denotes the complement (w.r.t.\ $\acts$) of a set $S\subseteq\acts$. Finally, $\textit{count\_bypasses}_p(b)$ tells
us that process $p$ loses at most $b$ competitions. It is expressed by
means of a largest fixed point construct: $\textit{count\_bypasses}_p(b)$ holds in all states such that
\begin{itemize}
\item after any action other than $\finishlock{q}$ for some $q\in\PID$, the same condition continues
  to hold, and
\item after any bypass by another process (i.e., after some $\finishlock q$ for some $q \neq p$), the same condition holds, but now with a budget $b{-}1$
  of one bypass less. Moreover, the formula evaluates to {\false} in case a bypass occurs when $b=0$.
\end{itemize}
For understanding the formula, it may help to add the conjunct $[\finishlock{p}]\textit{true}$, saying that when process $p$
finishes its \lock, there are no further requirements to check; however, this conjunct is vacuously true, so need not be included explicitly.

The formula for bounded bypass is exactly the same, but now fixing $\textit{doorway}_p$ to coincide
with the end of the first instruction of $p$'s \lock.

Using this formula, we can confirm the results on bounded bypass and post-doorway bounded bypass reported in this paper for Dekker's and Anderson's algorithms. 
For Dekker's, we place the $\textit{doorway}$ tag at the end of the first line (as we discussed in Remark~\ref{rem:DekkerPDBB}, this is the only possible placement of the doorway). It is not (currently) possible with mCRL2 to check that $\textit{pdbb}_p(b)$ is violated for all natural numbers $b$. However, for any bound we fill in (up to 100), bounded bypass is violated; moreover, we can replay in the tool the execution that witnesses that an arbitrary number of bypasses is possible (shown in Remark~\ref{rem:DekkerBB}).
For Anderson's algorithm, we place the $\textit{doorway}$ tag at the end of the fourth line and find that a bound of 1 is satisfied, but a bound of 0 is violated. We also find that the doorway can already end after the second line, but in that case a bound of 2 is needed.

We also confirm the result of Theorem~\ref{thm:bbp} for 2 and 3 processes: for both $n = 2$ and $n=3$, we find a tight bound of $n(n{-}1){-}1$ bypasses for a process that completes $\mathrm{W1}$.
Actually, in a preliminary version of this paper, we proved a $n(n{-}1)$ bound for Algorithm~\ref{alg} in the setting of atomic registers; it was the model checker that pointed out to us that such a bound was not tight and led us to reconsider our proofs, to tightly set the bound of Theorem~\ref{thm:bbp} at $n(n{-}1){-}1$.

The result of Theorem~\ref{thm:bbp} applies when atomic registers are used. As is stated in Proposition~\ref{pr:No Bounded Bypass} and Remark~\ref{rem:regbb}, this bound does not apply when safe or regular registers are used instead. Using the formula, we confirm that in those cases, no bound up to 100 is satisfied when there are 3 processes. For the specific case of 2 processes, we find that with both types of register a bound of 1 is sufficient. This is in accordance with the claim above Remark~\ref{rem:regular}.

\subsection{Intermittent Bounded Bypass}\mbox{}

\bigskip
\noindent
To capture intermittent bounded bypass, we need to make \df{intermittent} completely unambiguous, by deciding which assignments count as interrupting.
Naturally, an interrupting assignment may start before the process $p_i$ invokes its \lock, or end after $p_i$ reaches the corresponding critical section, as long as it overlaps with $p_i$'s \lock. The difficulty lays in the quantification over all acceptable definitions of the specific interrupting assignments. Following the application in the case study of \Sec{alg}, whether a write counts as an interrupting assignment may depend on the particular run in which it occurs. Here we opt for allowing \emph{any} choice of which writes in a run are classified as interrupting or non-interrupting.

We use the variables $b \mathbin\in\mathbbm{N}$ and $h\mathbin\in\mathbbm{N}$ to store the remaining budget of bypasses and interrupting assignments, respectively. Furthermore $W \subseteq \PID$ collects those processes that are currently engaged in a write operation. The formula $\textit{ibb}_\pid(h,b,W)$, defined below, holds in a given state of the mutual exclusion protocol with active writes $W\!$, when the process $\pid \in \PID$ will, in each possible execution continuing from that state, after any given occurrence of the action $\startlock{\pid}$, execute the action $\finishlock{\pid}$ after experiencing no more than $b$ bypasses, not counting those that occurred during a collection of no more than $h$ writes.
Thus, a mutual exclusion protocol satisfies intermittent bounded bypass of $f(n)$ within $h(n)$
interrupting assignments, for $n=|\PID|$, when the  initial state satisfies the formula
$\bigwedge_{\pid\in\PID}\textit{ibb}_\pid(h(n),f(n),\emptyset)$.

The formula $\textit{ibb}_\pid(h,b,W)$ is defined in four steps. In the first step, $\textit{ibb}_\pid(h,b,W)$ is defined in terms of an auxiliary formula $\textit{pre\_count}_\pid(h,b,W)$, which describes the intended behavior after the crucial action $\startlock{\pid}$ has occurred:
\begin{equation*}
\begin{array}{rl}
\textit{ibb}_\pid(h,b,W) :=
  \nu X(h,b,W)~. ~(
    &  \bigwedge_{q\in \PID} [\swby{q}] X(h,b,W\cup\{q\}) \\
   & \wedge\quad \bigwedge_{q\in \PID} [\fwby{q}] X(h,b,W{\setminus}\{q\}) \\
   & \wedge\quad [\overline{\swby{\it all} \cup \fwby{\it all}}]X(h,b,W)\\
   & \wedge\quad [\startlock{\pid}] \textit{pre\_count}_\pid(h,b,W))  \\[1ex]
\end{array}
\end{equation*}
In particular, $\textit{ibb}_\pid(h,b,W)$ holds for the largest set of states such that
\begin{itemize}
\item after the action $\startlock{\pid}$ occurs, the formula $\textit{pre\_count}_\pid(h,b,W)$ should hold, and
\item after any action $a \in \textit{acts}$ occurs, the same formula $\textit{ibb}_\pid(h,b,W')$ should still hold, where $W'$ is an updated version of $W$, by adding or removing a process from $W$ in case $a$ is the  start or finish of a write operation.
\end{itemize}
The second bullet point applies also when $a = \startlock{\pid}$, because $\textit{pre\_count}_\pid(h,b,W)$ should hold after \emph{any} occurrence of $\startlock{\pid}$, not merely after the first one.

Above we defined the \lock\ to start when its first instruction starts, but since in our model the start of the \lock\ is captured by the special actions $\startlock{\pid}$, our model creates space between the start of the lock and its first instruction. Hence, in our formula we explicitly have to skip over this gap before we start counting bypasses. This is the function of the second step in our definition of $\textit{ibb}_\pid(h,b,W)$, which defines $\textit{pre\_count}_\pid(h,b,W)$ in terms of  $\textit{skip\_first}(h,b,W)$, which is like $\textit{pre\_count}_\pid(h,b,W)$, but starting at the first action of  \lock:\vspace{-1ex}
\begin{equation*}
\begin{array}{rl}
\textit{pre\_count}_\pid(h,b,W) := \nu X(h,b,W)~.~(h{>}0\ \ \wedge
                &  \bigwedge_{q\in \PID{\setminus}\{\pid\}} [\swby{q}] X(h,b,W\cup\{q\}) \\
  \wedge & \bigwedge_{q\in \PID{\setminus}\{\pid\}} [\fwby{q}] X(h,b,W{\setminus}\{q\}) \\
  \wedge & [\actsexcept{\pid}{\setminus}(\swby{\it all} \cup \fwby{\it all})]X(h,b,W)\\
  \wedge & [\acts_{\pid}{\setminus} \swby{\pid}] \false \\
  \wedge & [\swby{\pid}] \textit{skip\_first}(h{-}1,b,W\cup \{p\}))
\end{array}
\end{equation*}
From our analysis we know that no protocol can satisfy intermittent bounded bypass, unless the first action of the \lock is a write, and this write counts as interrupting.
For this reason, our formula $\textit{pre\_count}_\pid(h,b,W)$ evaluates to {\false} in case $h=0$, and when the first action of $p$ encountered by \plat{$\textit{pre\_count}_\pid(h,b,W)$} is anything other than a write. Moreover, when encountering this first write we decrement $h$ by~$1$.

The third step in our definition defines $\textit{skip\_first}_\pid$ in terms of $\textit{count\_bypasses}_\pid$
in a similar vein, by not counting any bypasses until we reach the end of the first (write) action of $p$.
\begin{equation*}
\begin{array}{rl}
\textit{skip\_first}(h,b,W) := \nu X(h,b,W)~.~( &
  \bigwedge_{q\in \PID{\setminus}\{\pid\}} [\swby{q}] X(h,b,W\cup\{q\}) \\
  & \wedge\quad  \bigwedge_{q\in \PID{\setminus}\{\pid\}} [\fwby{q}] X(h,b,W{\setminus}\{q\}) \\
  & \wedge\quad [\overline{\swby{\it all} \cup \fwby{\it all}}]X(h,b,W)\\
  & \wedge\quad [\fwby{\pid}] \textit{count\_bypasses}_\pid(\mathit{init}(h,b,W{\setminus}\{p\})))\hspace{-5pt}
\end{array}
\end{equation*}
where $\mathit{init}(h,b,W)$ will be described in a moment.

The formula $\textit{count\_bypasses}_\pid$ keeps an open mind on whether any ongoing write shall be counted as interrupting or not. To this end it maintains as a parameter not just a single triple $(h,b,W)$ as used above, but a set $O$ of triples $(h,b,I)$, with $I$ denoting those writes that are currently active and are deemed interrupting. For each possible classification of encountered writes as interrupting or non-interrupting, we maintain such an option $(h,b,I) \in O$.
The formula $\textit{count\_bypasses}_\pid(O)$ expresses that process $\pid \in \PID$ will, in each possible execution continuing from the current state, for some  $(h,b,I)\in O$,  execute the action $\finishlock{\pid}$ after experiencing no more than $b$ bypasses, not counting those that occurred during a collection of at most $h$ writes chosen to be classified as interrupting, including the set $I$ of ongoing writes:
\begin{equation*}
\begin{array}{rl}
\textit{count\_bypasses}_\pid(O) := \nu X(O)~. ~ (O \mathbin{\neq} \emptyset & 
  \wedge \quad  \bigwedge_{q\in \PID} [\swby{q}] X(\mathit{wstart}(O, q)) \\
   & \wedge \quad \bigwedge_{q\in \PID} [\fwby{q}] X(\mathit{wfinish}(O, q))\\
 & \wedge \quad [\flexcept{\pid}] X(\mathit{bypass}(O))  \\
 & \wedge \quad [\overline{\swby{\it all} \cup \fwby{\it all} \cup \finishlock{\it all}}]X(O))
\end{array}
\end{equation*}
\noindent
Here we use the following functions:
\[\begin{array}{l@{~:=~}l}
\mathit{init}(h,b,W) & opt(\{(h{-}|I|,b,I) \mid I \subseteq W \wedge h {-} |I| \geq 0\})\\
\mathit{wstart}(O,q) & opt(O \cup \{(h{-}1,b,I\cup\{q\}) \mid (h,b,I) \in O \wedge h{>}0\})\\
\mathit{wfinish}(O,q) & opt(\{(h,b,I{\setminus}\{q\}) \mid (h,b,I) \in O\})\\
\mathit{bypass}(O) & opt(\{(h,b,I) \mid (h,b,I) \in O \wedge I\neq \emptyset\}
                  \cup \{(h,b{-}1,\emptyset) \mid (h,b,\emptyset) \in O \wedge b{>}0\})
\end{array}\]
\noindent where
\begin{align*}
    \mathit{opt}(O) :=& \{(h,b,I) \in O \mid \neg\exists (h', b', I')\in O.(h \leq h' \land b \leq b' \land I \subseteq I' \land (h', b', I')\neq (h,b,I))\}
\end{align*}

When $\textit{count\_bypasses}_\pid$ is called, one option $(h,b,I) \in O$ is created for each choice of a set $I\subseteq W$ of the ongoing writes as being interrupting; this is captured by the function \textit{init}. Each time we encounter a start-write, say by process $q\in \PID$, we double the amount of options by allowing this write to be interrupting or not, unless we have exhausted our budget of interrupting assignments. When we encounter a bypass (an action in $\flexcept{\pid}$), we decrement the bound $b$ only for those options $(h,b,I)\in O$ in which no interrupting assignment is active, that is, $I=\emptyset$.  Such an option is deleted in case the budget $b$ is already exhausted. 
Finally, in case any action occurs that is not a start or finish write, or a finish-lock, we simply keep checking the same formula.
The first clause $O \mathbin{\neq} \emptyset$ of $\textit{count\_bypasses}_\pid(O)$ says that the formula evaluates to \textit{false} when we run out of valid options.

To improve performance when model checking, we call the function \textit{opt} whenever we create a new set $O$, which removes non-optimal options from the set; an option is non-optimal when another option in the set is strictly better in one dimension  (more budget for interrupting assignments, more budget for bypasses, or a superset of the set of interrupting assignments) and at least as good in the others.

We use this formula to check the claims made in \Sec{algSafeRegs} for Algorithm~\ref{alg} with Lamport's One-Bit protocol as DLF.

With safe registers, we find that the bound of $n^2 -2$ bypasses within $n$ interrupting assignments reported in Theorem~\ref{thm:bib} is indeed satisfied and tight for 3 processes.
With 2 processes both the bound on bypasses and the bound on interrupting assignments can be decremented by one, in accordance with our observations at the end of Section~\ref{sec:alg} and in Section~\ref{sec:modelcheckingbb} above.


With regular registers, the results are exactly the same as with safe ones: the bound of $n^2 -2$ bypasses within $n$ interrupting assignments is satisfied and tight for 3 processes, while the bound can be lowered by 1 in both directions for 2 processes.

As is noted in \Sec{ibb}, intermittent bounded bypass coincides with bounded bypass when there is a bound on interrupting assignments of exactly 1, namely the first instruction of the {\lock}.
As such, we can also run the intermittent bounded bypass formula on the model with atomic registers, using a bound of 1 for interrupting assignments and of $n(n-1)-1$ for bypasses.
We find that these bounds are satisfied and tight for both 2 and 3 processes.

To verify the correctness of the proof of Theorem~\ref{thm:bib}, we also developed an alternative
formula for intermittent bounded bypass that specifically classifies as interrupting those writes that are treated as such in that proof, namely the ones indicated in the text below Remark~\ref{rem:regbb}---this is different from the approach presented above, where every possible choice of at most $h(n)$ interrupting assignments is allowed. That formula is included in \cite{GGSzenodo} for the interested reader, but not here. When checked on the safe register model, it confirms the bounds reported in Theorem~\ref{thm:bib}, and as such confirms that these writes are indeed the $h(n)$ interrupting assignments that are needed for this proof.

The process of developing this formula helped us refine the definition of which writes should be counted as interrupting in Algorithm~\ref{alg}.
In an earlier proposal, instead of requiring conditions (a), (b) and (c) to be satisfied at some moment \emph{during} the execution of the write, it was required that they be satisfied at the \emph{end} of the write. 
In the corresponding earlier version of our alternative formula,
this resulted in too few writes being classified as interrupting, and hence the bound on bypasses being violated. The violating traces produced by the model checker for this version of the property helped identify and fix the issue.

\section{Discussion and Conclusions}
\label{sec:disc}

In this paper, we provided a few notions of bounded bypass as liveness properties that, as Raynal advocates, strengthen the notion of starvation-freedom. In doing this, we corrected two mistakes in the definition provided in \cite{Raynal} by explicitly requiring that: (1) the process that invokes \lock\ eventually enters the critical section, and (2) this must happen after a bounded number of bypasses
{\em after the moment in which the process has completed the execution of the first instruction of its \lock}. Without the first requirement, bounded bypass does not imply starvation-freedom; without the second amendment, we argued that no {\MUTEX} protocol satisfies bounded bypass.

Next, we presented two weakenings of bounded bypass, namely post-doorway and intermittent bounded bypass, that are half-way properties between starvation-freedom and bounded bypass. 
Subsequently, we proved that every deadlock-free protocol can be turned into one that satisfies (some form of) bounded bypass, both in the setting of safe, regular and atomic registers.

In addition to developing new definitions, we also prove that they are (or are not) satisfied for three relevant algorithms.
Building on this, we take a lesson from Lamport, who stated in \cite{LamportMutex2}: {\it `Our proofs have been done in the style of standard ``journal mathematics'', using informal reasoning that in principle can be reduced to very formal logic, but in practice never is. [...] The behavioral reasoning used in our correctness proofs, and in most other published correctness proofs of concurrent algorithms, is inherently unreliable; we advise the reader to be skeptical of such proofs.'} 
Thus, to support the abstract reasoning provided in the proofs, we formalized the definitions in a formal logic and used model checking to confirm our claimed results.
We thereby aim to reduce reader's skepticism and improve developer's confidence.

Of particular note is our verification of Algorithm~\ref{alg}.
The starvation-freedom of this algorithm was automatically verified in \cite{GLS25}, by using Lamport's One-Bit protocol \cite{LamportMutex2} as DLF and by running the resulting code for 3 processes.
There, it was shown that the resulting protocol behaves as well as Lamport's Three-Bit protocol \cite{LamportMutex2}, as it satisfies starvation-freedom for all the three models considered in this paper (safe/regular/atomic registers). 
Here, we went one step further: we provided formulae for the different types of bounded bypass defined in this paper and used them to confirm all claims made here. This also had a positive impact on the theoretical development of the work, since it identified for the protocol of Algorithm \ref{alg} a non-tight bound (in the setting of atomic registers) and an imprecise identification of the interrupting assignments (in the setting of safe registers); this allowed us to fine-tune the theoretical part accordingly.

\bibliographystyle{alphaurl}
\bibliography{bibliography.bib}

\end{document}